\documentclass[11pt,a4paper]{article}
\usepackage[T1]{fontenc}
\usepackage[utf8]{inputenc}
\usepackage{lmodern}
\usepackage[margin=2.7cm]{geometry}
\usepackage{amsmath,amssymb,mathtools,amsthm,url}
\usepackage{microtype}
\usepackage[hidelinks]{hyperref}
\hypersetup{
  pdftitle={CCZ-Equivalence and Enumeration of Triprojective APN Functions},
  pdfauthor={Jingchuan Ma; Yanhua Liu; Qiaoyun Huang},
  pdfsubject={Preprint},
  pdfkeywords={APN, CCZ-equivalence, EA-equivalence, finite fields}
}

\DeclareMathOperator{\GL}{GL}
\DeclareMathOperator{\Aut}{Aut}

\newcommand{\Ftwo}{\mathbb F_2}
\newtheorem{theorem}{Theorem}
\newtheorem{lemma}{Lemma}
\newtheorem{proposition}{Proposition}

\title{\bfseries CCZ-Equivalence and Enumeration of\\Triprojective APN Functions}
\author{%
  Jingchuan Ma\thanks{Email: \href{mailto:kiciot@qq.com}{kiciot@qq.com}. ORCID: \href{https://orcid.org/0009-0001-3703-471X}{0009-0001-3703-471X}.}
  \quad Yanhua Liu\thanks{Email: \href{mailto:Lyhwa@fzu.edu.cn}{Lyhwa@fzu.edu.cn}. ORCID: \href{https://orcid.org/0000-0002-6076-9968}{0000-0002-6076-9968}.}
  \quad Qiaoyun Huang\thanks{Corresponding author. Email: \href{mailto:huangqy@fdzcxy.edu.cn}{huangqy@fdzcxy.edu.cn}. ORCID: \href{https://orcid.org/0009-0008-2122-7307}{0009-0008-2122-7307}.}\\[0.65em]
  \normalsize Fuzhou University Zhicheng College
}
\date{}
\begin{document}
\maketitle
\begin{abstract}
We classify all admissible parameters of the G\"olo\u{g}lu--K\"olsch
triprojective construction of almost perfect nonlinear (APN) functions
up to Carlet--Charpin--Zinoviev (CCZ) equivalence. The construction has three coefficients in $K=\mathbb F_{2^m}$
and a Frobenius exponent $k$ coprime to $m$. For every $m>1$, we give a
necessary-and-sufficient criterion that includes changes of both the
coefficients and the exponent. Each parameter choice determines one of
the two irreducible cubic polynomials over $\mathbb F_2$. Two functions
are equivalent precisely when their exponents and cubics agree, or when
the exponents are negatives modulo $m$ and the cubics are reciprocal.
When $7\nmid m$, every admissible member is equivalent to a Li--Kaleyski
representative: the larger coefficient space adds no CCZ classes.
When $7\mid m$, those representatives are inadmissible; we construct
replacements and classify the full family. There are exactly $\varphi(m)$
CCZ classes, extending the previously known count for the older
subfamilies to all admissible coefficients and all degrees.
We also count the parameter triples and compute canonical representatives
and explicit equivalence maps in deterministic polynomial time for every
degree, including degrees divisible by seven. The proof combines a reduction to classical semilinear conjugacy
with an intrinsic recovery of the scalar field from the polar bilinear map.
The reduction from arbitrary EL equivalence is algebraic in every degree,
including degrees three and six.
\end{abstract}
\section{Introduction and Main Result}
Almost perfect nonlinear (APN) functions attain differential uniformity
two, the minimum possible over a vector space over $\mathbb F_2$.
They are studied both as cryptographic functions and through their
connections with codes~\cite{CarletCharpinZinoviev1998}.
For a parameterized construction, however, different coefficients need
not give inequivalent functions. Carlet--Charpin--Zinoviev
(CCZ) equivalence preserves differential uniformity and nonlinearity;
classification under this relation distinguishes genuinely different
members of an APN construction~\cite{CarletCharpinZinoviev1998,Kaleyski2022}.

The triprojective construction of G\"olo\u{g}lu and
K\"olsch~\cite{GologluKolschTriprojective2026} allows three coefficients
over an extension field and contains the earlier Li--Kaleyski
families~\cite{LiKaleyski2024}. We give a complete CCZ-equivalence
classification of all admissible choices of coefficients and Frobenius
exponent, together with exact parameter counts and explicit equivalence
maps. For $7\nmid m$, we reduce every admissible coefficient triple to
an earlier Li--Kaleyski representative, proving that the enlarged
coefficient space introduces no additional CCZ classes. For $7\mid m$,
we construct admissible replacements and complete the classification.
The result concerns this parameterized family, not all APN functions
in dimension $3m$.

\subsection{The Family and Its Equivalence Problem}
Throughout, $K=\mathbb F_q$, $q=2^m$, $m>1$, $n=3m$, and
$\sigma(x)=x^{2^k}$, where $1\leq k<m$ and $\gcd(k,m)=1$.
Vectors are columns, and a choice of basis identifies $K^3$ with
$\Ftwo^n$. A function $F:\Ftwo^n\to\Ftwo^n$ is APN if
$F(x+a)+F(x)=b$ has at most two solutions for every $a\ne0$ and every
$b$. It is quadratic if each coordinate has algebraic degree at most two
in its algebraic normal form over $\Ftwo$. For a quadratic map fixing zero, the polar derivative
$D_aF(x)=F(x+a)+F(x)+F(a)$ is linear over $\Ftwo$ in each argument;
the APN condition is equivalent to $\ker D_aF=\{0,a\}$ for $a\ne0$.
Frobenius automorphisms on vectors and matrices act entrywise;
$\sigma^j$ denotes composition $j$ times. Exponents such as
$t^{\sigma^2+\sigma+1}=t^{2^{2k}+2^k+1}$ are interpreted accordingly.
Unless stated otherwise, matrix ranks are over $K$.
For $p=(a,b,c)\in K^*\times K^2$, define the map
$F_p^\sigma=(F_{p,1}^\sigma,F_{p,2}^\sigma,F_{p,3}^\sigma):K^3\to K^3$ by
\begin{align}
F_{p,1}^{\sigma}(x,y,z)&=x^{\sigma+1}+ay^\sigma z+bx^\sigma y+cx^\sigma z,
\nonumber\\
F_{p,2}^{\sigma}(x,y,z)&=ay^{\sigma+1}+z^\sigma x+bz^\sigma y+cx^\sigma y,
\label{eq:family}\\
F_{p,3}^{\sigma}(x,y,z)&=z^{\sigma+1}+x^\sigma y.
\nonumber
\end{align}
We call the parameter triple admissible when
\begin{equation}
 a t^{\sigma^2+\sigma+1}+b t^{\sigma+1}+ct+1\ne0
 \quad\text{for every }t\in K.\label{eq:admissible}
\end{equation}
Under these hypotheses, the construction is quadratic APN;
for odd $m$ it is a permutation~\cite[Theorems 1.3 and 1.6]{GologluKolschTriprojective2026}.

To state the classification, associate to $p$ the matrices
\begin{equation}
 M_p=\begin{pmatrix}0&1&0\\b^\sigma&0&1\\a+cb^\sigma&0&c\end{pmatrix},
 \qquad N_p=M_pM_p^\sigma\cdots M_p^{\sigma^{m-1}},
 \label{eq:matrix}
\end{equation}
\begin{equation}
 f_p(t)=\det(tI-N_p),\qquad
 f_0(t)=t^3+t+1,\quad f_1(t)=t^3+t^2+1.
 \label{eq:invariant}
\end{equation}
The dependence on $k$ is understood in $M_p,N_p,f_p$; we write
$f_p^\sigma$ when it must be explicit. The map $H_p(x)=M_px^\sigma$
is $\sigma$-semilinear, meaning that it is additive and satisfies
$H_p(\lambda x)=\lambda^\sigma H_p(x)$. Its $m$th iterate is the
$K$-linear map with matrix $N_p$. Lemma~\ref{lem:dempwolff-specialization}
will show that every admissible parameter has $f_p\in\{f_0,f_1\}$.
For a monic cubic with constant term one, its reciprocal is
$f^*(t)=t^3f(1/t)$; in particular, $f_0^*=f_1$ and $f_1^*=f_0$.

We distinguish equivalence over $K$ from equivalence over $\Ftwo$.
Linear equivalence without a field qualifier means
$G(Mx)=L(F(x))$ for $M,L\in\GL(3m,2)$.
Extended-linear (EL) equivalence allows an added $\Ftwo$-linear map:
$G(Mx)=L(F(x))+R(x)$. Affine equivalence allows affine permutations
on the input and output, without an added function of the input.
Extended-affine (EA) equivalence allows an additional affine function
of the input.
CCZ-equivalence means affine equivalence of the graphs
in the product of the input and output spaces. Equivalence under
$K$-linear, respectively semilinear, input and output changes means
$M,L\in\GL(3,K)$, respectively $M,L\in\Gamma L(3,K)$.
Here $\Gamma L(3,K)$ consists of the invertible maps $x\mapsto Tx^\tau$
with $T\in\GL(3,K)$ and $\tau\in\Aut(K)$.
For quadratic functions fixing zero, EA-equivalence is equivalent to
EL-equivalence: an input translation changes the quadratic map by a
linear map and a constant; the former is absorbed into $R$, and evaluation
at zero removes the latter. For quadratic APN functions, Yoshiara's
theorem also identifies CCZ-equivalence with EA-equivalence
\cite[Theorem~1]{Yoshiara2012}.

\begin{theorem}[Classification of admissible parameters]
\label{thm:main}
Let $F_p^\sigma,F_{p'}^{\sigma'}$ be admissible members over the same $K$.
For every $m>1$, the following are equivalent:
\begin{enumerate}
\item the functions are CCZ-equivalent;
\item they are EA-equivalent;
\item they are EL-equivalent;
\item they are linearly equivalent over $\Ftwo$;
\item they are affinely equivalent;
\item either $k'=k$ and $f_{p'}=f_p$, or
      $k'\equiv-k\pmod m$ and $f_{p'}=f_p^*$.
\end{enumerate}
Consequently, there are exactly $\varphi(m)$ CCZ classes, where $\varphi$
is Euler's totient function, as $k$ and the
admissible coefficients vary. For each fixed $k$ and $m>2$, there are
exactly two CCZ classes represented at that exponent, and each contains
$(q^3-1)/7$ admissible triples with that fixed $k$.
For $m=2$ all admissible members are linearly equivalent over $\Ftwo$,
although they form two classes under $K$-linear changes.
\end{theorem}

For $m>2$, pairing $(k,f)$ with $(m-k,f^*)$ identifies the
$2\varphi(m)$ fixed-exponent labels in pairs, yielding $\varphi(m)$
global CCZ classes.

\paragraph{Proof Strategy.}
There are two distinct equivalence problems to solve. First, we express
$F_p^\sigma$, after a $K$-linear input change, as $Q_{H_p}(x)=x\times H_p(x)$.
Conjugate semilinear maps give equivalent quadratic functions, but the
converse requires proof: independent input and output changes must force
conjugacy, up to a removable scalar. Theorem~\ref{thm:klinear} establishes
this converse. Dempwolff's classification~\cite{Dempwolff2010} then supplies
the two cubic polynomials and identifies a centralizer of order seven.
Second, an arbitrary $\Ftwo$-linear equivalence need not preserve the
chosen $K$-space structure. We determine the derivation space of the polar
bilinear map and recover its scalar field. At degree three a second field
occurs in that derivation space; its derivative coefficient ranks exclude
it as the image of the scalar field. This proves an algebraic reduction
to $\Gamma L(3,K)$ for every $m>2$, including degrees three and six.
The parameter count and explicit representatives follow
from the same classification. Section~\ref{sec:canonicalization} makes
the norm construction deterministic and computes a canonical representative
and explicit equivalence witness in polynomial time for every $m>1$.

\subsection{Relation to Previous Work}
G\"olo\u{g}lu and K\"olsch proved admissibility, APN behavior, inclusions
of earlier families, and inequivalence to Gold functions. These construction and inclusion results do not decide equivalence
for every pair of admissible coefficient triples.
Shi, Peng, Kan, and Gao~\cite[Theorem 11]{ShiPengKanGao2025}
identify the two Li--Kaleyski constructions~\cite{LiKaleyski2024} under
inverse exponents. Their Theorem 17 and Corollary 1 state the
$\varphi(m)$ CCZ class count for
$m>2$, $m\notin\{4,6\}$ and $7\nmid m$ in the APN case.
Their Theorem 8 assumes three pairwise distinct exponents up to sign;
Remark 2 describes a version without this restriction, and Section 4
treats the equal-exponent case for those fixed coefficient functions.
Our main addition is the necessary-and-sufficient
reduction of every admissible $(a,b,c)$, its exact parameter count, and
the inclusion of $7\mid m$ and all small degrees.
When $7\nmid m$, the older representatives already exhaust the classes
of the larger family; its coefficient parameters add no CCZ classes.

The biprojective precedents already establish parameter classification,
class counts and automorphism methods. G\"olo\u{g}lu classifies
$(q,q)$-biprojective APN functions on $K^2$~\cite{GologluClassification2023}.
K\"olsch~\cite[Theorems 2.4 and 3.6]{KolschExtension2024} removes the
coefficient parameter in another biprojective extension, identifies inverse
exponents, and obtains $\varphi(m)/2$ classes. Its Theorem~3.6 treats
$m>2$, $m\notin\{4,6\}$; the discussion following that theorem also
settles the small cases.
Kaspers and Zhou give exact equivalence criteria and counts for the
Pott--Zhou and Taniguchi families~\cite{KaspersZhouLower2022,KaspersZhouExponential2021}.
These results concern different functions on $K^2$.

The earlier group-theoretic method comes from biprojective equivalence
\cite{GologluKolschBiprojective2025} and was extended to triprojective
functions in~\cite[Lemmas 5 and 6, Theorem 8]{ShiPengKanGao2025}.
Proposition~\ref{prop:stabilizer} calculates the $K$-linear stabilizer
through the associated irreducible semilinear map. Our unrestricted
reduction instead recovers the scalar field from the polar bilinear map,
so it does not require a primitive prime divisor or a Sylow argument.
Restricting the Frobenius exponents of coordinate functions does
not restrict a $K$-linear change of coordinates to a permutation matrix
with diagonal scalings. For a concrete example, take
$K=\Ftwo[\theta]/(\theta^5+\theta^2+1)$ and $\sigma(x)=x^2$.
At this same exponent, the functions
in~\cite[Definitions~2.1 and~3.1]{BartoliStanica2026} are
\begin{align*}
 G_a(x,y,z)&=(x^3+ax^2z+yz^2,\ x^2z+y^3,\ xy^2+ay^2z+z^3),\\
 H_b(x,y,z)&=(x^3+bxy^2+yz^2,\ xy^2+z^3,\ x^2z+y^3+by^2z).
\end{align*}
Encode $\sum_{j=0}^4 r_j\theta^j$ by the integer
$\sum_{j=0}^4 r_j2^j$, with $r_j\in\{0,1\}$; entries below are field
elements in this encoding, not integers modulo $32$.
For $a=9=\theta^3+1$ and $b=1$, set
\[
 U=\begin{pmatrix}19&18&22\\25&8&1\\1&6&12\end{pmatrix},
 \qquad
 V=\begin{pmatrix}12&21&1\\6&1&8\\1&26&25\end{pmatrix}.
\]
Both matrices have determinant one, and direct substitution gives
\[
 H_1(X)=V\,G_9(UX)\qquad\text{for every }X\in K^3.
\]
The polynomials $t^7+9t+1$ and $t^7+t+1$ have no root in $K$, so
both parameters satisfy the admissibility condition in~\cite{BartoliStanica2026}.
This gives a same-exponent counterexample to
\cite[Theorem~4.5, version~1]{BartoliStanica2026} at $m=5$.

Orthoderivatives provide established invariants for quadratic APN equivalence
\cite{CanteautCouvreurPerrin2022,Kaleyski2022}. Their differential spectra
give independent finite checks of the algebraic classification in degrees
three and six.

\section{Classification over the Extension Field}
\label{sec:classification-algebra}

We first classify the admissible parameters under $K$-linear input and
output changes.  The main step is a converse: an equivalence between the
quadratic maps determines a conjugacy between the associated semilinear
maps, up to a scalar that can be removed.  The classification of the
semilinear maps themselves is classical~\cite{Dempwolff2010}.

\subsection{Admissibility and Semilinear Conjugacy}
Write $p=(a,b,c)$ and let $M_{p}$,
$H_{p}(x)=M_{p}x^\sigma$, and $f_{p}$ be as in
\eqref{eq:matrix} and \eqref{eq:invariant}.  When the automorphism needs to
be specified, we write $f_{p}^{\sigma}$.  All parameter triples
in this section are admissible unless stated otherwise.
Set
\[
 C_p=\begin{pmatrix}0&1&0\\0&0&1\\1&b&0\end{pmatrix},
 \qquad Q_H(x)=x\times H(x),
\]
where $\times$ is the ordinary cross product on $K^3$; its coordinates
use addition in place of subtraction in characteristic two.
Conjugacy of semilinear maps always means conjugacy by a $K$-linear
map. In matrices, $G=T^{-1}HT$ is the relation
$M_G=T^{-1}M_HT^\sigma$.

\begin{lemma}
\label{lem:irreducible}
Let $a\ne0$.  The admissibility polynomial has no root in $K$ if and only
if $H_{p}$ is irreducible, that is, it has no nonzero proper
invariant $K$-subspace.  Moreover,
\begin{equation}
 F_{p}^{\sigma}(C_{p}^{-1}x)
 =Q_{H_{p}}(x)=x\times H_{p}(x).
 \label{eq:cross-change}
\end{equation}
\end{lemma}

\begin{proof}
Put
\[
 D_{p}=
 \begin{pmatrix}0&0&1\\1&0&0\\c&a&0\end{pmatrix}.
\]
Expanding the cross product gives
$F_{p}^{\sigma}(X)=(C_{p}X)\times
(D_{p}X^\sigma)$.  Since
$\det C_{p}=1$, $\det D_{p}=a$, and
$D_{p}(C_{p}^{-1})^\sigma=M_{p}$,
this proves \eqref{eq:cross-change} and invertibility of $H_{p}$.

An invariant $K$-line is equivalent to
$D_{p}X^\sigma=\lambda C_{p}X$ for some $X\ne0$ and
$\lambda\in K^*$.  The first two equations are
$z^\sigma=\lambda y$ and $x^\sigma=\lambda z$.  They force $x\ne0$.
After rescaling $x=1$, they give
$z\ne0$, $\lambda=z^{-1}$ and $y=z^{\sigma+1}$; the third equation
is exactly
\[
 az^{\sigma^2+\sigma+1}+bz^{\sigma+1}+cz+1=0.
\]
Conversely, these expressions construct an invariant line from any root.

It remains to exclude an invariant plane when there is no invariant line.
If $W$ were an invariant $K$-plane, the induced map on $K^3/W$ would have
the form $t\mapsto d t^\sigma$ with $d\in K^*$.  The map
$t\mapsto t^{\sigma-1}$ permutes $K^*$, because its kernel is
$\operatorname{Fix}(\sigma)^*=\Ftwo^*=\{1\}$.  Thus the induced map
has a nonzero fixed vector and its difference from the identity is not
surjective.  Consequently $H_{p}-I$ is not surjective on $K^3$.
As an $\Ftwo$-linear endomorphism of $K^3$, it then has a nonzero kernel.  A fixed
vector spans an invariant $K$-line, a contradiction.
\end{proof}

\begin{lemma}[Specialization of Dempwolff's classification]
\label{lem:dempwolff-specialization}
Let $H(x)=Mx^\sigma$ be irreducible on $K^3$, and put
\[
 N=M M^\sigma\cdots M^{\sigma^{m-1}}.
\]
The characteristic polynomial $f_H(t)=\det(tI-N)$ is either
$f_0(t)=t^3+t+1$ or $f_1(t)=t^3+t^2+1$.
For a fixed $\sigma$, two such maps are conjugate under $\GL(3,K)$
if and only if their polynomials $f_H$ agree, and both possibilities occur.
Furthermore,
\begin{equation}
 \{T\in\operatorname{End}_K(K^3):TH=HT\}
 =\Ftwo[H^m]\cong\mathbb F_8.
 \label{eq:commuting-field}
\end{equation}
\end{lemma}

\begin{proof}
These are the specializations of
\cite[Ths.~2.4, 2.10 and 3.4]{Dempwolff2010} to fixed field $\Ftwo$
and $K$-dimension three.  The commuting endomorphism ring is a finite
division ring by irreducibility, hence a field.  The cited results identify
it with $\Ftwo[H^m]$, give dimension $m$ for $K^3$ over this field, and
classify $H$ by the irreducible minimal polynomial of $H^m$ over $\Ftwo$.
This field therefore has eight elements and that minimal polynomial has
degree three.  Every generating element of $\mathbb F_8/\Ftwo$ occurs
in the classification, since the finite-field norm used in its construction
is surjective.

To identify the characteristic polynomial over $K$, use the identity
\[N=M N^\sigma M^{-1}.\]
Its characteristic polynomial therefore has
coefficients in $\operatorname{Fix}(\sigma)=\Ftwo$.  It has degree
three and annihilates $N$, so it equals the irreducible minimal polynomial
just described.  The two monic irreducible cubics over $\Ftwo$ are
$f_0$ and $f_1$.
\end{proof}

The irreducibility of $f_H$ in this lemma is over $\Ftwo$.
When $3\mid m$, the cubic splits over $K$; this does not contradict
irreducibility of the semilinear map $H$, which concerns invariant
$K$-subspaces rather than factors of $f_H$ over $K$.

\subsection{From Input and Output Equivalence to Conjugacy}
Conjugacy gives an equivalence by the cross-product identity below.
The converse is the essential step: the two independent linear changes
in an equivalence cannot be chosen freely.

\begin{theorem}
\label{thm:klinear}
Fix $\sigma$.  Two admissible maps $F_{p}^{\sigma}$ and
$F_{p'}^{\sigma}$ are equivalent under invertible $K$-linear
input and output maps if and only if
$f_{p}^{\sigma}=f_{p'}^{\sigma}$.
The criterion is unchanged if an arbitrary $\Ftwo$-linear summand is
allowed in the equivalence.
\end{theorem}

\begin{proof}
For $T\in\GL(3,K)$, the cross product satisfies
\begin{equation}
 (Tu)\times(Tv)=\det(T)T^{-T}(u\times v).
 \label{eq:cross-covariance}
\end{equation}
Hence $G=T^{-1}HT$ implies
\[
 Q_H(Tx)=\det(T)T^{-T}Q_G(x).
\]
Lemma~\ref{lem:dempwolff-specialization} and
\eqref{eq:cross-change} prove sufficiency.

For necessity, suppose
\begin{equation}
 Q_G(x)=P Q_H(Tx)+R(x),
 \qquad P,T\in\GL(3,K),
 \label{eq:restricted-el}
\end{equation}
where $R$ is $\Ftwo$-linear.  Write $H(x)=M_Hx^\sigma$ and
$G(x)=M_Gx^\sigma$.  Let $[x]_\times$ denote the matrix of
$y\mapsto x\times y$.  The polar derivative of $Q_H$ is
\[
 D_xQ_H(y)=[x]_\times M_Hy^\sigma+[H(x)]_\times y.
\]
For $x\ne0$, its $y^\sigma$ coefficient has rank two and left kernel
$Kx$.  The coefficients of $y^\sigma$ and $y$ are unique: evaluating
a zero relation at $y=t e_i$ and using $\sigma\ne\operatorname{id}$
forces both coefficient matrices to vanish.
Polarizing \eqref{eq:restricted-el} therefore gives
\[
 [x]_\times M_G=P[Tx]_\times M_H T^\sigma,
 \qquad Kx=P^{-T}K(Tx).
\]
A $K$-linear map preserving every $K$-line is scalar, as follows by
applying it to basis vectors and their pairwise sums.  Thus
\begin{equation}
 P^{-T}T=dI,\qquad P=d^{-1}T^T
 \label{eq:equivalence-transpose}
\end{equation}
for some $d\in K^*$.  Substituting
\eqref{eq:cross-covariance} in the coefficient equality gives
\[
 [x]_\times\bigl(M_G-rT^{-1}M_H T^\sigma\bigr)=0
 \quad\text{for all }x,\qquad r=\det(T)/d.
\]
Each column of the parenthesized matrix has cross product zero with
every $x$, so the matrix vanishes.  Consequently
\begin{equation}
 G=rT^{-1}HT.
 \label{eq:scalar-conjugacy}
\end{equation}
Every scalar $r$ can be removed by conjugacy: choose $\gamma\in K^*$
with $\gamma^{\sigma-1}=r$ and use
$(\gamma I)^{-1}H(\gamma I)=rH$.
Thus $G$ and $H$ are conjugate and have the same polynomial.
The resulting quadratic maps already satisfy
$Q_G(x)=P Q_H(Tx)$, so $R=0$ as well.
\end{proof}

\begin{proposition}
\label{prop:stabilizer}
Let $H$ be irreducible with associated automorphism $\sigma$, and put
$\mathcal C=C_{\GL(3,K)}(H)$.  Then $|\mathcal C|=7$.
The complete set of $K$-linear pairs satisfying $Q_H(Tx)=A Q_H(x)$ is
\begin{equation}
 \begin{split}
 T&=\gamma C,\\
 A&=\gamma^{\sigma+1}\det(C)C^{-T},
 \end{split}
 \qquad \gamma\in K^*,\quad C\in\mathcal C.
 \label{eq:stabilizer-pairs}
\end{equation}
In particular, this group has order $7(q-1)$.  Allowing an
$\Ftwo$-linear summand with these $K$-linear input and output maps
does not enlarge it.
\end{proposition}

\begin{proof}
Equation~\eqref{eq:commuting-field} gives $|\mathcal C|=7$.
The necessity argument in Theorem~\ref{thm:klinear} shows that every
possible input map conjugates $H$ to a scalar multiple of $H$.
Removing that scalar as above yields $T=\gamma C$.
This representation is unique, since a scalar map commutes with $H$
exactly when its scalar lies in $\Ftwo^*=\{1\}$.
Substitution in \eqref{eq:cross-covariance} gives the displayed $A$.
 The output map is unique.  Indeed, a $K$-linear functional annihilating
all values of $Q_H$ annihilates its polar derivatives, and uniqueness of
the two coefficients forces it into every left kernel $Kx$ of
$[x]_\times M_H$.  It is therefore zero, so the values of $Q_H$ span
$K^3$ over $K$.
The zero-summand conclusion follows from Theorem~\ref{thm:klinear}.
\end{proof}

\begin{theorem}
\label{thm:semilinear}
Two admissible maps $F_{p}^{\sigma}$ and
$F_{p'}^{\sigma'}$ are equivalent under invertible semilinear
input and output maps over $K$ if and only if
\begin{align}
 &\sigma'=\sigma,\qquad
 f_{p'}^{\sigma'}(t)=f_{p}^{\sigma}(t),
 \label{eq:direct-semilinear}\\
 \text{or}\quad
 &\sigma'=\sigma^{-1},\qquad
 f_{p'}^{\sigma'}(t)=t^3f_{p}^{\sigma}(1/t).
 \label{eq:inverse-semilinear}
\end{align}
An additional $\Ftwo$-linear summand does not change this criterion.
The stabilizer under semilinear input and output maps has order
$7m(q-1)$.
\end{theorem}

\begin{proof}
Write the equivalence in the form
$Q_G=P\circ Q_H\circ T+R$, with
$T(x)=T_0x^\phi$, $P(z)=P_0z^\theta$, and $R$ linear over $\Ftwo$.
The matrices $T_0,P_0$ are invertible over $K$, and
$\phi,\theta\in\Aut(K)$. Set $\eta=\phi\theta$.
In the transformed polar derivative the powers of the derivative variable
are $\eta$ and $\eta\sigma$, with nonzero rank-two coefficient matrices
in each nonzero direction.  In the target they are
$\operatorname{id}$ and $\sigma'$.
Distinct Frobenius maps of $K$ are linearly independent over $K$, by
uniqueness of the linearized polynomial representation.  Therefore
\begin{equation}
 \{\eta,\eta\sigma\}
 =\{\operatorname{id},\sigma'\}.
 \label{eq:frobenius-support}
\end{equation}
If $\eta=\operatorname{id}$, then $\sigma'=\sigma$.
After raising coefficients to a field automorphism, the remaining
equivalence is $K$-linear.  This coefficient operation preserves the
polynomial $f_H$, whose coefficients lie in $\Ftwo$, so
Theorem~\ref{thm:klinear} gives \eqref{eq:direct-semilinear}.

Otherwise \eqref{eq:frobenius-support} forces
$\eta=\sigma'=\sigma^{-1}$.
The identity
\begin{equation}
 Q_H(H^{-1}x)=H^{-1}x\times x=Q_{H^{-1}}(x)
 \label{eq:inverse-map-equivalence}
\end{equation}
reduces this case to the first one.  The $m$th power of $H^{-1}$ is
$(H^m)^{-1}$, and its characteristic polynomial is the reciprocal
$t^3 f_H(1/t)$.  This proves necessity.  Conjugacy and
\eqref{eq:inverse-map-equivalence} also prove sufficiency.
The argument uses polar derivatives until the reduction to
Theorem~\ref{thm:klinear}, so it applies with an added linear summand;
that theorem then forces the summand to vanish.

For the stabilizer count, write a pair as $A\circ Q_H=Q_H\circ T$.
In this convention $A=P^{-1}$, so $A$ and $T$ have the same field
automorphism in the direct case.  Each of the $m$ field
automorphisms occurs, because coefficient conjugation preserves $f_H$
and hence preserves its $K$-conjugacy class.  For each such automorphism,
the number of pairs is $7(q-1)$ by Proposition~\ref{prop:stabilizer}.
The inverse case cannot stabilize the same map: it exchanges the two
distinct reciprocal cubics, including when $m=2$ and
$\sigma=\sigma^{-1}$.  This proves the stated order.
\end{proof}

\subsection{Parameter Counts and Representatives}
We now recover the actual coefficient triples from a semilinear
conjugacy class. Counting triples requires this additional step; it is
not the same as counting conjugacy classes.

\begin{proposition}
\label{prop:parametercount}
For every fixed $\sigma$, there are exactly two $K$-linear equivalence
classes of admissible triples, each of size $(q^3-1)/7$.
Their parameters and equivalence witnesses can be obtained as follows.
For $H=H_{p}$ and $v\ne0$, put
\begin{equation}
 B_v=[v,Hv,H^2v],\qquad
 (u_0,u_1,u_2)^T=B_v^{-1}H^3v,\qquad
 p_v=(u_0,u_1^{\sigma^{-1}},u_2^{\sigma^{-2}}).
 \label{eq:cyclic-parameter-action}
\end{equation}
Write $p_v=(a_v,b_v,c_v)$. Then $p_v$ runs through exactly the
class of $p$.
Define
\begin{equation}
 W_c=\begin{pmatrix}0&0&1\\0&1&c^\sigma\\1&c&cc^\sigma\end{pmatrix},
 \qquad T_v=B_vW_{c_v}^{-1}.
 \label{eq:cyclic-conjugator}
\end{equation}
The corresponding equivalence in the original coordinates is
\begin{equation}
 F_{p_v}^{\sigma}(X)
 =\det(T_v)^{-1}T_v^T
 F_{p}^{\sigma}
 \bigl(C_{p}^{-1}T_v C_{p_v}X\bigr).
 \label{eq:parameter-equivalence-witness}
\end{equation}
As $\sigma$ also varies, there are $\varphi(m)$ semilinear equivalence
classes when $m>2$, and one when $m=2$, where $\varphi$ denotes Euler's
totient function.
\end{proposition}

\begin{proof}
Every nonzero vector is cyclic for an irreducible semilinear map:
if its iterates first become dependent before spanning $K^3$, their span
is a nonzero proper invariant subspace.  Thus $B_v$ is invertible.
The coefficient $u_0$ in \eqref{eq:cyclic-parameter-action} is nonzero,
since otherwise $\langle Hv,H^2v\rangle_K$ would be invariant.
Direct multiplication gives, for $H_{a,b,c}$,
\[
 H^3e_3=a e_3+b^\sigma He_3+c^{\sigma^2}H^2e_3,
 \qquad [e_3,He_3,H^2e_3]=W_c.
\]
Consequently $H T_v=T_vH_{p_v}$.
The latter map is irreducible and its parameters are admissible by
Lemma~\ref{lem:irreducible}.  Equations~\eqref{eq:cross-change} and
\eqref{eq:cross-covariance} now give
\eqref{eq:parameter-equivalence-witness}.
The same construction, starting with any irreducible $\sigma$-semilinear
map, represents it by a matrix $M_{a,b,c}$.  Both conjugacy classes in
Lemma~\ref{lem:dempwolff-specialization} therefore occur among the
admissible parameters.

Put $\mathcal C=C_{\GL(3,K)}(H)$, the order-seven group of
Proposition~\ref{prop:stabilizer}. Two vectors $v,w\ne0$ give the same triple in
\eqref{eq:cyclic-parameter-action} if and only if $w=Cv$ for some
$C\in\mathcal C$.  Commutation proves one implication.  For the other,
the $K$-linear map taking $H^iv$ to $H^iw$ for $i=0,1,2$ commutes with
$H$ by their identical recurrence.  This action of $\mathcal C$ is
free: a commuting map fixing $v$ fixes its cyclic basis.  Each parameter
triple consequently has exactly seven preimages, proving the count.
Conversely, conjugating any other representative into $H$ and applying
its conjugator to $e_3$ shows that it occurs in this construction.
Theorem~\ref{thm:klinear} identifies these two parameter classes with
the $K$-linear equivalence classes.

There are two classes for each of the $\varphi(m)$ possible $\sigma$.
For $m>2$, Theorem~\ref{thm:semilinear} pairs
$(\sigma,f)$ with $(\sigma^{-1},t^3f(1/t))$ and makes no further
identifications.  For $m=2$, the two cubics are exchanged at the same
$\sigma$.  This gives the final assertion.
\end{proof}

Two useful explicit parts of this action are scalar and inverse changes.
Rescaling a cyclic vector by $\lambda\in K^*$ sends
\begin{equation}
 (a,b,c)\longmapsto
 \bigl(a\lambda^{\sigma^3-1},
 b\lambda^{\sigma^2-1},c\lambda^{\sigma-1}\bigr).
 \label{eq:scalar-parameter-action}
\end{equation}
Indeed, $H^i(\lambda v)=\lambda^{\sigma^i}H^iv$ in its recurrence.
Writing $\tau=\sigma^{-1}$, applying $H^{-3}$ to the recurrence for
$e_3$ gives the following representative for $H^{-1}$:
\begin{equation}
 \sigma\longmapsto\tau,\qquad
 (a,b,c)\longmapsto
 \bigl(a^{-\tau^3},c/a^{\tau^2},b/a^\tau\bigr).
 \label{eq:inverse-parameter-action}
\end{equation}
Its equivalence witness is obtained from the cyclic basis of $H^{-1}$
as in \eqref{eq:cyclic-conjugator}, together with
\eqref{eq:inverse-map-equivalence}.
Applying a field automorphism to all three coefficients leaves the
polynomial $f_{p}$ unchanged.

\begin{proposition}
\label{prop:representatives}
If $7\nmid m$, the two triples
\begin{equation}
 (1,1,0),\qquad (1,0,1)
 \label{eq:constant-representatives}
\end{equation}
represent the two $K$-linear classes for every $k$ coprime to $m$.
Their polynomial labels are respectively $f_0,f_1$ if
$m\bmod 7\in\{1,2,4\}$, and $f_1,f_0$ if
$m\bmod 7\in\{3,5,6\}$.
If $7\mid m$, both triples in \eqref{eq:constant-representatives}
are inadmissible, but two representatives can be constructed by
finite-field norms and \eqref{eq:cyclic-parameter-action}.
\end{proposition}

\begin{proof}
The two matrices in \eqref{eq:matrix} obtained from
\eqref{eq:constant-representatives} have entries in $\Ftwo$,
characteristic polynomials $f_0,f_1$, and order seven.
The $m$th iterates of the associated semilinear maps therefore have matrices $M^m$.
If $7\nmid m$, these powers still have distinct irreducible cubics over
$\Ftwo$.  Dempwolff's converse criterion
\cite[Cor.~2.5]{Dempwolff2010} shows that the semilinear maps are
irreducible: $\Ftwo[H^m]$ is a field of order eight and $K^3$ has
dimension $m$ over it.  Lemma~\ref{lem:irreducible} proves admissibility.
The labels follow from the two Frobenius orbits
$\{1,2,4\}$ and $\{3,5,6\}$ of nonzero exponents modulo seven.
If $7\mid m$, then $H^m=I$, which is incompatible with irreducibility
by Lemma~\ref{lem:dempwolff-specialization}.  Lemma~\ref{lem:irreducible}
then proves inadmissibility.

Here is a construction that also covers $7\mid m$.
Choose $\omega\in\mathbb F_8\setminus\Ftwo$, and use
$\omega$ and $\omega^3$ to obtain the two polynomials.
If $3\mid m$, choose $a\in K^*$ with
$N_{K/\mathbb F_8}(a)=\omega$ and take $(a,0,0)$.
For its map $H$, the $m$th power has the three eigenvalues
$\omega^{\sigma^2},\omega^\sigma,\omega$, because $H^3$ has
diagonal coefficients $a^{\sigma^2},a^\sigma,a$ and companion
automorphism $\sigma^3$.  Since $k$ is coprime to $m$, these are the
three conjugates of $\omega$ over $\Ftwo$.
The same converse criterion proves irreducibility and admissibility.

If $3\nmid m$, put $E=\mathbb F_{2^{3m}}$ and choose $r$ with
$r\equiv k\pmod m$ and $r\equiv0\pmod3$.
The automorphism $\alpha:z\mapsto z^{2^r}$ of $E$ has order $m$,
fixed field $\mathbb F_8$, and restriction $\sigma$ to $K$.
Choose $w\in E^*$ with $N_{E/\mathbb F_8}(w)=\omega$ and define
$H(z)=w\alpha(z)$ on the three-dimensional $K$-space $E$.
Then $H^m$ is multiplication by $\omega$, so its polynomial over
$\Ftwo$ is the desired irreducible cubic.  Again $H$ is irreducible.
Taking $v=1$ in \eqref{eq:cyclic-parameter-action} now gives an
admissible triple in the original parameterization; its cyclic basis
is invertible by Proposition~\ref{prop:parametercount}.
The required norm preimages exist because norms between finite fields
are surjective.  Replacing $\omega$ by $\omega^3$ exchanges the cubics.
\end{proof}

For $7\nmid m$, Proposition~\ref{prop:representatives} and
Theorem~\ref{thm:klinear} therefore reduce every admissible coefficient
triple to one of the two earlier representatives. G\"olo\u{g}lu and
K\"olsch~\cite[Propositions 3.2 and 3.4]{GologluKolschTriprojective2026}
identify them, by explicit linear changes, with the Li--Kaleyski
subfamilies. Thus the full coefficient space adds no CCZ classes in
these degrees. This assertion concerns every admissible triple, beyond
the fixed-coefficient functions classified in~\cite{ShiPengKanGao2025}.

The same conjugacy equations also produce an equivalence for a given
pair of parameters, without enumerating their class.

\begin{proposition}
\label{prop:equivalence-algorithm}
Given two admissible parameter triples with the same $\sigma$ and the
same polynomial $f_{p}$, an explicit $K$-linear equivalence
can be obtained by solving
\begin{equation}
 M_{p}T^\sigma=T M_{p'}
 \label{eq:conjugator-system}
\end{equation}
as a homogeneous linear system over $\Ftwo$ in $9m$ variables.
Its solution space has dimension three over $\Ftwo$, and every
nonzero solution is invertible.  Substituting such a solution for $T_v$
in \eqref{eq:parameter-equivalence-witness} gives the equivalence.
\end{proposition}

\begin{proof}
Conjugacy gives an invertible solution $T_0$.
Every solution is $T_0C$ with $C$ commuting with $H_{p'}$,
so \eqref{eq:commuting-field} gives precisely an $\mathbb F_8$-space
of solutions, of dimension three over $\Ftwo$.  Alternatively, the kernel
and image of a nonzero solution are invariant under the corresponding
irreducible maps, proving its invertibility directly.
The identity follows from \eqref{eq:cross-covariance} and
\eqref{eq:cross-change}.
\end{proof}

Thus the criterion consists of constant-size matrix products over $K$
to compute $f_{p}$, followed, when a witness is required, by
\eqref{eq:conjugator-system}.  The inverse-automorphism case is handled
first by \eqref{eq:inverse-parameter-action}.  These are constructive
consequences of the classification.

\section{Uniform Reduction of Arbitrary EL Equivalence}
\label{sec:el}
We recover the scalar field from the polar bilinear map. This avoids the
primitive-prime and Sylow-subgroup exceptions in degrees three and six.
Throughout this section $m>2$, and $H=M\sigma$ is irreducible on $K^3$.
Put $N=H^m$ and
\[
 b_H(x,y)=x\times H(y)+y\times H(x),\qquad
 \mathcal C_H=\Ftwo[N]\cong\mathbb F_8.
\]
The commuting algebra is identified by
Lemma~\ref{lem:dempwolff-specialization}.

For a binary bilinear map $b:V\times V\to W$, define its derivations by
\[
 \operatorname{Der}(b)=
 \{(D,E): b(Dx,y)+b(x,Dy)=E b(x,y)\text{ for all }x,y\},
\]
where $D\in\operatorname{End}_{\Ftwo}(V)$ and
$E\in\operatorname{End}_{\Ftwo}(W)$.
Let $\mathcal D(b)$ be its projection to the first coordinate.
If $b_G(Sx,Sy)=T b_F(x,y)$, then
\begin{equation}
 \mathcal D(b_G)=S\mathcal D(b_F)S^{-1}.
 \label{eq:derivation-covariance}
\end{equation}
Thus this space is intrinsic to the polar bilinear map under unrestricted
input/output equivalence.

\begin{lemma}[The derivation space]
\label{lem:derivation-space}
Every derivation of $b_H$ is $K$-linear in both coordinates, and
\[
 \mathcal D(b_H)=K I_3+\mathcal C_H.
\]
More precisely its pairs have the form
\begin{align*}
 D&=aI_3+C,\\
 E&=(a+a^\sigma+\operatorname{tr}C)I_3+C^T,
 \qquad a\in K,\quad C\in\mathcal C_H.
\end{align*}
The binary dimension is $m+2$.
\end{lemma}
\begin{proof}
Write the unique linearized expressions, with coefficient matrices over $K$,
\[
 D=\sum_{i=0}^{m-1}D_i\sigma^i,\qquad
 E=\sum_{j=0}^{m-1}E_j\sigma^j.
\]
In the derivation identity, the terms contributed by $D_i$ have
unordered pairs of Frobenius indices
\[
 \{i,1\},\quad\{0,i+1\},
\]
whereas those contributed by $E_j$ have pair $\{j,j+1\}$.
Indices are modulo $m$. Distinct pairs can be compared separately,
by uniqueness of polynomial functions in the independent coordinates
of $x,y$, with each coordinate degree less than $|K|$.

If $m\ge5$, the pair $\{i,1\}$ for $i\notin\{0,1,2\}$ is neither
an adjacent pair nor a pair containing zero. Its coefficient is
$(D_i u)\times Mv$, so $D_i=0$.
The pair $\{0,2\}$ then forces $D_1=0$, and $\{0,3\}$ forces
$D_2=0$. The last two coefficients have the form
$u\times MD_i^\sigma v$; their vanishing forces $D_i=0$.

For clarity, the two short cycles admit equally algebraic arguments.
We use the elementary fact that a $K$-linear map preserving every
$K$-line is scalar. At $m=4$, the same nonadjacent pairs first give
$D_1=D_3=0$. Comparing the pair $\{1,2\}$ gives
\[
 [D_2u]_\times M=E_1[M^\sigma u]_\times.
\]
Apply both sides to $M^\sigma u$. The scalar-line fact implies
$D_2=tMM^\sigma$ for some $t\in K$.
The pair $\{0,3\}$ gives
\[
 [x]_\times MD_2^\sigma=E_3[M^{\sigma^3}x]_\times.
\]
Applying both sides to $M^{\sigma^3}x$ shows that
$MD_2^\sigma M^{\sigma^3}=t^\sigma N$ is scalar. Since $N$ has
the squarefree cubic characteristic polynomial $f_0$ or $f_1$, it is
not scalar. Hence $t=0$.

At $m=3$, the diagonal pair $\{1,1\}$ says
$(D_1u)\times Mv+(D_1v)\times Mu=0$.
Testing basis vectors after the substitution $(Mu,Mv)$ gives
$D_1=tM$. Comparing $\{0,2\}$ and applying its matrix identity to
$M^{\sigma^2}x$ then makes $t^\sigma N$ scalar, so $D_1=0$.
The diagonal pair $\{0,0\}$ says $MD_2^\sigma$ is scalar.
The pair $\{1,2\}$, as above, gives $D_2=tMM^\sigma$.
Together these make $t^\sigma N$ scalar and force $D_2=0$.
No finite enumeration is used in either short-cycle argument.

Only $D_0$ remains. The adjacent pairs on the right are distinct when
$m>2$, so all $E_j$ with $j\ne0$ vanish as well: the values of the
corresponding cross-product coefficient span $K^3$.
Thus $D,E$ are $K$-linear. Comparing the $y^\sigma$ coefficient gives
\[
 [Dx]_\times+[x]_\times MD^\sigma M^{-1}=E[x]_\times.
\]
The infinitesimal cross-product identity is
\[
 [Dx]_\times+[x]_\times D
   =((\operatorname{tr}D)I_3+D^T)[x]_\times.
\]
Consequently $[x]_\times A=J[x]_\times$, where
$A=D+MD^\sigma M^{-1}$ and
$J=E+(\operatorname{tr}D)I_3+D^T$.
Applying this to $x$ makes $A$ scalar, say $A=cI_3$;
the spanning property then gives $J=cI_3$.
Taking traces yields
$c=\operatorname{tr}D+(\operatorname{tr}D)^\sigma$.
With $a=\operatorname{tr}D$, the matrix $C=D+aI_3$ commutes with $H$,
so $C\in\mathcal C_H$. This gives the displayed form, and direct
substitution proves its converse.
Finally $KI_3\cap\mathcal C_H=\Ftwo I_3$, because a scalar commutes
with $H$ precisely when it is fixed by $\sigma$. The dimension is $m+3-1$.
\end{proof}

\begin{lemma}[Fields in the derivation space]
\label{lem:derivation-fields}
If $m>2$ and $m\ne3$, the only unital subfield of order $2^m$ contained
in $\mathcal D(b_H)$ is $KI_3$. If $m=3$, there are exactly two such
subfields: $KI_3$ and $\mathcal C_H$.
\end{lemma}
\begin{proof}
If $3\nmid m$, the generated algebra $K[\mathcal C_H]=K[N]$ is the
field of order $2^{3m}$, since $f_H$ remains irreducible over $K$.
That field has a unique subfield of order $2^m$.

If $3\mid m$, diagonalize $N$ over $K$. Its three eigenvalues are
the binary conjugates of a root of $f_H$, so
\[
 \mathcal D(b_H)=
 \{\operatorname{diag}(a+c,a+c^2,a+c^4):a\in K,\ c\in\mathbb F_8\}.
\]
A unital subfield of order $2^m$ in this diagonal algebra has the form
\[
 \{\operatorname{diag}(z,z^{2^r},z^{2^s}):z\in K\},
 \qquad 0\le r,s<m,
\]
since each coordinate projection is a field isomorphism onto $K$.
For it to lie in the displayed derivation space, the image of
$z\mapsto z+z^{2^r}$ must be contained in
$\{c+c^2:c\in\mathbb F_8\}$, a binary space of dimension two.
For $r\ne0$ that image has dimension $m-\gcd(r,m)\ge m/2$.
Thus if $m\ge6$, both $r$ and $s$ are zero.

For $m=3$, write $A=z+z^{2^r}$ and $B=z+z^{2^s}$.
Membership requires and is equivalent to $B=A+A^2$.
For $r=0$ this gives $s=0$; for $r=1$ it gives $s=2$.
For $r=2$ it would require $z^{2^s}=z+z^2+z^4$, which is not a field
automorphism. The two possibilities are exactly the stated subfields.
\end{proof}

The extra field at $m=3$ can be distinguished by its coefficient ranks.
This comparison uses two known presentations of the same triprojective
map, not a classification of arbitrary quadratic maps.

\begin{lemma}[The alternative field in degree three]
\label{lem:alternative-field}
Let $m=3$. In the homogeneous quadratic presentation induced by
the commuting field $\mathcal C_H$, both derivative coefficient
matrices have rank three at a nonzero direction. Hence that field
cannot be the image of the scalar field under an EL equivalence
between triprojective members.
\end{lemma}
\begin{proof}
By Proposition~\ref{prop:representatives} and $K$-linear conjugacy,
assume $H=M\sigma$ with $M$ one of the two constant matrices of order
seven. Identify its three-dimensional binary space with a field
$C\cong\mathbb F_8$ so that $M$ is multiplication by
$\theta\in C\setminus\Ftwo$. Then $\mathcal C_H=\Ftwo[M]$, and
the input space is $K\otimes_{\Ftwo}C$.
In trace-dual output coordinates the binary cross product on $C$ is
\[
 c\times d=c^2d^4+c^4d^2.
\]
Indeed its pairing with a third vector is the nonzero alternating
Moore determinant, which identifies the usual cross product in any
binary basis with its trace-dual output basis.

View the tensor product instead as a three-dimensional space over $C$,
and apply fourth-power Frobenius to its output coordinates over $C$.
The original map scales by $c^{-1}$ under $c\in C^*$, so after this
change it scales by $c^{-4}=c^3$. Binary quadraticity and the distinct
scalar weights imply purity with associated automorphism $d\mapsto d^2$.
At the direction $1\otimes1$, its derivative, evaluated on
$b\otimes d$ with $b\in K,d\in C$, is
\[
 d^2(\theta^2 b^\sigma+\theta b)
   +d(\theta b^\sigma+\theta^2 b).
\]
This follows by expanding the cross product; such elementary tensors
span the space over $C$.
The automorphism $\sigma$ of $K=\mathbb F_8$ has characteristic
polynomial $t^3+1$ on its binary space: it has order three and a
one-dimensional fixed space. Therefore the two coefficient determinants,
over $C$, are both
\[
 \det(\theta^2\sigma+\theta I)
 =\det(\theta\sigma+\theta^2 I)
 =\theta^3+\theta^6\ne0.
\]
The inequality follows from the exact order seven of $\theta$.

We derive the coefficient law directly. In source cross-product coordinates
$F(x)=x\times M_Fx^{2^\kappa}$, $\kappa\in\{1,2\}$; this $K$-linear change
preserves the scalar field and ranks. Here $\kappa$ is the source exponent. Let $G:C^3\to C^3$ be the target
in the alternative coordinates above, so $G(cz)=c^3G(z)$. Write
\[
 b_F(x,y)=A_xy^{2^\kappa}+B_xy,\qquad
 b_G(z,w)=\widetilde A_z w^2+\widetilde B_z w.
\]
Here $A_x=[x]_\times M_F$ and $B_x=[M_Fx^{2^\kappa}]_\times$ have rank
two for $x\ne0$. Suppose an EL equivalence $G(Sx)=TF(x)+R(x)$
conjugates the input field $KI_3$ to the target commuting field.
With all coordinate changes absorbed, conjugation gives a field
isomorphism $\phi:K\to C$ and $S(ax)=\phi(a)Sx$.
Polarization removes $R$ and gives $b_G(Sx,Sy)=Tb_F(x,y)$;
no semilinearity of $T$ is assumed.

Put $e=1+2^\kappa$. The binary span of the values of $b_F$ is a
$K$-subspace, because $b_F(ax,ay)=a^e b_F(x,y)$ and $\gcd(e,7)=1$.
A $K$-linear row $\ell$ annihilating this span satisfies
$\ell[x]_\times M_F=0$ for every $x$, by independence of
$y\mapsto y^{2^\kappa}$ and $y\mapsto y$. This independence follows by
restricting to $y=te_i$, with $t=1$ and $t\notin\Ftwo$.
Since $M_F$ is invertible and basis-vector cross products span $K^3$,
$\ell=0$. Thus the values of $b_F$ span the entire binary output.

Apply the polarized identity to $(ax,ay)$ and use target homogeneity.
For $u=b_F(x,y)$ this gives $T(a^e u)=\phi(a)^3T(u)$, which extends
to every $u$ by spanning. Since $3e^{-1}\equiv1,2\pmod7$ for
$e=3,5$, respectively, we obtain
\[
 T(tu)=\phi(t)^{2^\delta}T(u),\qquad
 \delta=0\ (\kappa=1),\quad \delta=1\ (\kappa=2).
\]
Hence $Sx=P\phi(x)$ and $Tu=Q\phi(u)^{2^\delta}$ with
$P,Q\in\GL_3(C)$. Field maps and matrix powers below are entrywise.
Putting $v=\phi(y)$ in the polarized identity gives
\[
 \widetilde A_{Sx}P^2v^2+\widetilde B_{Sx}Pv
 =Q\phi(A_x)^{2^\delta}v^{2^{\kappa+\delta}}
   +Q\phi(B_x)^{2^\delta}v^{2^\delta}.
\]
Independence of $v\mapsto v^2$ and $v\mapsto v$, as above, and
$v^8=v$ give the explicit transformation law
\[
\begin{array}{c|cc}
 &\widetilde A_{Sx}&\widetilde B_{Sx}\\\hline
 \kappa=1&Q\phi(A_x)(P^2)^{-1}&Q\phi(B_x)P^{-1}\\
 \kappa=2&Q\phi(B_x)^2(P^2)^{-1}&Q\phi(A_x)^2P^{-1}.
\end{array}
\]
Invertible left/right factors and field isomorphisms preserve ranks.
The rank pair is therefore preserved for $\kappa=1$ and exchanged for $\kappa=2$.
Since $S$ is bijective, it is $(2,2)$ at every nonzero target direction,
contradicting the previously derived pair $(3,3)$ at $1\otimes1$.
\end{proof}

\begin{proposition}[Uniform unrestricted reduction]
\label{prop:elreduction}
For every $m>2$, any EL equivalence between admissible triprojective
members has zero added linear term and both input and output maps in
$\Gamma L(3,K)$. If their exponents are $k,k'$, then $k'=k$ or
$k'\equiv-k\pmod m$.
\end{proposition}
\begin{proof}
Work first in the cross-product coordinates of
Lemma~\ref{lem:irreducible}; their changes are $K$-linear.
Polarizing an EL equivalence gives $b_G(Sx,Sy)=T b_F(x,y)$.
By \eqref{eq:derivation-covariance} and
Lemma~\ref{lem:derivation-fields}, its input map $S$
normalizes $KI_3$ if $m\ne3$. At $m=3$, the only other possibility
is excluded by Lemma~\ref{lem:alternative-field}.
Thus $S\in\Gamma L(3,K)$ in every degree under consideration.
Write $S=S_0\tau_j$ with $\tau_j(x)=x^{2^j}$, and write the output
and added linear maps as $T=\sum_hT_h\tau_h$, $R=\sum_tR_t\tau_t$.

Substitute $\lambda x$ into $G(Sx)=T(F(x))+R(x)$.
Its left side has scalar weight $2^j(1+2^{k'})$ modulo $2^m-1$;
the output terms have weights $2^h(1+2^k)$; and the linear terms
have weights $2^t$. The first weights have two ones in their cyclic
binary representations and the last have one. For $m>2$ the output
weights are distinct. Independence as functions on $K^*$ gives $R=0$
and at most one nonzero $T_h$, using again that $F(K^3)$ spans $K^3$
over $K$. Invertibility forces one such term and
\[
 \{j,j+k'\}=\{h,h+k\}\pmod m.
\]
Thus $k'=k,h=j$ or $k'=-k,h=j+k'$, and $T$ is semilinear.
\end{proof}

\begin{proof}[Completion of Theorem~\ref{thm:main}]
For $m>2$, Proposition~\ref{prop:elreduction} and the semilinear
classification give exactly the stated criterion, without exceptional
finite computations. Conversely the constructed equivalences are binary
linear. CCZ-equivalence implies EA-equivalence by
Yoshiara's theorem~\cite{Yoshiara2012}, and quadratic translation reduces
EA to EL. Linear equivalence implies affine equivalence, which implies
EA-equivalence. For $m=2$, the explicit inverse-semilinear equivalence
exchanges the two cubics at the same exponent, so all admissible members
are linearly equivalent. This completes the theorem for every $m>1$.
\end{proof}

\section{Exact Verification}
\label{sec:verification}
The classification is proved algebraically in Section~\ref{sec:el}.
Independent exact implementations check the parameter formulas, canonical
representatives, and explicit witnesses.
The new canonicalization implementation covers degrees
$2,3,4,5,6,7,14,21,42,49$: 150 witnesses pass 276,576 evaluations on
zero, binary basis vectors and basis-pair sums. These determine the entire
quadratic function identity for each selected witness. Degree 49 exercises
two seventh-root levels, and degrees 21 and 42 exercise the branch
$3\mid m$. The code also checks deterministic selection, canonical equality
on the tested equivalence orbits, idempotence and invalid-input rejection.
The finite checks validate the specified instances; they are not inputs
to the all-degree classification or complexity proofs.

\section{Deterministic Representatives and Equivalence Witnesses}
\label{sec:canonicalization}

The classification gives a deterministic polynomial-time representative and
an explicit equivalence witness for every degree $m>1$. Fix a supplied monic irreducible polynomial
$h\in\Ftwo[z]$ of degree $m$, realize $K=\Ftwo[z]/(h)$, and order its
elements by the integers represented by their polynomial-basis coefficients.
The input is an admissible $(k,p)$ in this fixed encoding. The output is
$(m,h,k_0,p_0)$; equality of outputs is understood within the same encoding.

We first make the norm construction in Proposition~\ref{prop:representatives}
deterministic. The required root selection is a special case of classical
deterministic factoring in small characteristic~\cite{Shoup1991}; the
following argument suffices for the completely split polynomials used here.

\begin{lemma}[Deterministic norm preimages]
\label{lem:deterministic-norm}
Let $E\supset\mathbb F_8$ be an explicitly represented finite field,
$d=[E:\mathbb F_8]$, and let $\omega\in\mathbb F_8^*$ have order seven.
One can deterministically construct $w\in E$ with
$N_{E/\mathbb F_8}(w)=\omega$ in time polynomial in $[E:\Ftwo]$.
The construction requires neither a primitive element nor a discrete logarithm.
\end{lemma}
\begin{proof}
Put
\[
 s=v_7(d),\qquad c=(8^d-1)/7^{s+1},\qquad u=c^{-1}\pmod7.
\]
The lifting-the-exponent identity gives $v_7(8^d-1)=s+1$, so $u$ exists.
Start with $w_0=\omega^u$ and, for $1\le j\le s$, choose a root $w_j$
of $t^7-w_{j-1}$. Inductively $w_{j-1}$ has order $7^j$.
The cyclic subgroup of $E^*$ of order $7^{s+1}$ contains all seven roots,
each of order $7^{j+1}$. These roots are nonzero and the polynomial has
derivative $t^6$, so it is squarefree. The final element satisfies
\[
 N_{E/\mathbb F_8}(w_s)
 =w_s^{(8^d-1)/7}=w_s^{7^sc}=(\omega^u)^c=\omega.
\]
This also covers $s=0$, with no root extraction.

Here is a deterministic selection rule. For a monic squarefree polynomial
$f\in E[t]$ of degree $e\le7$ that splits completely, the algebra
$E[t]/(f)$ is a product of $e$ copies of $E$.
The binary-linear map $z\mapsto z^2+z$ therefore has kernel
$\Ftwo^e$. If $e>1$, a binary kernel basis contains an element
$z\notin\{0,1\}$. Its values at the roots include both zero and one,
so $\gcd(f,z)$ is a proper nonconstant factor. Repeating on this monic
factor reaches a linear polynomial.
Order coordinates by increasing power of $t$, then increasing bit of the
fixed field basis. Process binary matrix columns in that order, use the
highest nonzero row as pivot, and choose the first resulting kernel basis
vector different from 0 and 1. Always retain the monic gcd factor.
These conventions fix every choice.

Each binary system has dimension at most $7[E:\Ftwo]$, and at most six
proper-factor steps select one root. There are $s=O(\log d)$ seventh-root
levels. Binary elimination and fixed-degree polynomial arithmetic are
polynomial-time operations in the supplied field encoding. More explicitly,
with schoolbook field arithmetic, one splitting stage costs
$O([E:\Ftwo]^3)$ bit operations, including exponentiation-based inverses.
\end{proof}

First compute $f_p$ from \eqref{eq:matrix}--\eqref{eq:invariant}.
For $m>2$, put
\begin{equation}
 k_0=\min\{k,m-k\},\qquad
 \varepsilon=\begin{cases}0,&k<m/2,\\1,&k>m/2,\end{cases}
 \qquad
 g=\begin{cases}f_p,&\varepsilon=0,\\f_p^*,&\varepsilon=1.\end{cases}
 \label{eq:canonical-pair}
\end{equation}
If $7\nmid m$, choose the unique triple $p_0$ among $(1,1,0)$ and
$(1,0,1)$ whose polynomial for exponent $k_0$ is $g$.
If $7\mid m$, use the following deterministic version of
Proposition~\ref{prop:representatives}:
\begin{itemize}
\item If $3\mid m$, split $g$ over $K$ by the selection rule in
Lemma~\ref{lem:deterministic-norm}, obtaining a root $\omega$.
Use that lemma with $E=K$ and $d=m/3$ to choose $a$ of norm $\omega$,
and take $p_0=(a,0,0)$.
\item If $3\nmid m$, use the cubic extension $E=K[t]/(g)$,
with ordered $K$-basis $(1,t,t^2)$ and the supplied binary basis of $K$.
The polynomial $g$ remains irreducible over $K$, and its residue
$\omega=t$ generates the subfield $\mathbb F_8$.
Use the lemma with $d=m$ to choose $w$ of norm $\omega$.
Choose $r$ modulo $3m$ with $r\equiv k_0\pmod m$ and $r\equiv0\pmod3$,
and apply \eqref{eq:cyclic-parameter-action} to $H(z)=wz^{2^r}$ with $v=1$.
The resulting triple is $p_0$.
\end{itemize}
For $m=2$, instead set
$k_0=1$ and $p_0=(1,1,0)$, and set $\varepsilon=1$ precisely when
$f_p\ne f_{p_0}$.

Here is a deterministic witness for the equivalence to $p_0$.
Write $\rho(x)=x^{2^{k_0}}$ and put
$H_s=H_p$ if $\varepsilon=0$, and $H_s=H_p^{-1}$ otherwise.
The associated automorphism of $H_s$ is $\rho$; its matrix $M_s$ is
$M_p$ in the first case and $(M_p^{-1})^{\sigma^{-1}}$ in the second.
Solve the linear system over $\Ftwo$
\begin{equation}
 M_sT^\rho=TM_{p_0},\qquad D=(\det T)^{-1}T^T,
 \label{eq:canonical-system}
\end{equation}
choosing a nonzero solution $T$ deterministically. For example, order its
$9m$ binary variables by row, column, and increasing basis index.
Gaussian elimination over $\Ftwo$, processing columns in this order,
selects the relation given by the first dependent column.
Define $V=D$ and
\begin{equation}
 U(x)=
 \begin{cases}
 C_p^{-1}TC_{p_0}x,&\varepsilon=0,\\
 C_p^{-1}M_s(TC_{p_0})^\rho x^\rho,&\varepsilon=1.
 \end{cases}
 \label{eq:canonical-witness}
\end{equation}

\begin{proposition}
\label{prop:canonicalization}
The procedure terminates for every admissible input. Two inputs over the
fixed encoded field give the same tuple $(m,h,k_0,p_0)$ if and only if
their functions are CCZ-equivalent. The computed maps $U,V$ are invertible
and $\Ftwo$-linear, and satisfy
\begin{equation}
 F_{p_0}^{\rho}(x)=V F_p^\sigma(U(x)).
 \label{eq:canonical-equivalence}
\end{equation}
For every $m>1$, the representative and witness are computable in
deterministic polynomial time. Counting additions, multiplications and
squarings in $K$ as field operations, computing $(k_0,g)$ costs $O(m^2)$
field operations, as does selecting the representative when $7\nmid m$.
Witness recovery costs a further $O(m^2)$ field operations and $O(m^3)$
bit operations. With schoolbook arithmetic, the entire procedure,
including the case $7\mid m$, has the conservative bit bound
$O(m^4\log m)$ and polynomial storage.
\end{proposition}

\begin{proof}
First, the polynomial label provides an admissibility test.
For any $a\ne0$, $M_p$ is invertible and
$N_p=M_pN_p^\sigma M_p^{-1}$, so $f_p\in\Ftwo[t]$.
If $f_p$ is either irreducible cubic, its minimal polynomial over $\Ftwo$
equals $f_p$. Thus $S=\Ftwo[N_p]\cong\mathbb F_8$ and
$\dim_S K^3=m$. Dempwolff's converse criterion
\cite[Cor.~2.5]{Dempwolff2010} implies that $H_p$ is irreducible;
Lemma~\ref{lem:irreducible} gives admissibility. The converse follows from
Lemma~\ref{lem:dempwolff-specialization}.
Proposition~\ref{prop:representatives} supplies the two specified constant
triples when $7\nmid m$. For $7\mid m$, Lemma~\ref{lem:deterministic-norm}
supplies the required norm preimages. In the cubic-extension branch,
$\gcd(r,3m)=3$, so $z\mapsto z^{2^r}$ has order $m$, fixes
$\mathbb F_8$, and restricts to $\rho$ on $K$.
Thus $H^m$ is multiplication by $\omega$, with characteristic polynomial
$g$ over $K$. In the other branch the three eigenvalues of $H_{p_0}^m$
are the binary conjugates of $\omega$, again giving $g$.
The irreducibility and cyclic-vector arguments in
Proposition~\ref{prop:representatives} prove admissibility in both cases.

For $m>2$, Theorem~\ref{thm:main} identifies CCZ classes precisely by
the two-element orbits of $(k,f)\mapsto(m-k,f^*)$.
Equation~\eqref{eq:canonical-pair} selects the same $(k_0,g)$ exactly on
each such orbit; the subsequent choice depends only on this pair and the
fixed encoding. Distinct pairs cannot produce the same representative,
since its exponent and polynomial recover the pair. The case $m=2$
follows from the unique class in Theorem~\ref{thm:main}.

The maps $H_s$ and $H_{p_0}$ have the same automorphism and polynomial.
Proposition~\ref{prop:equivalence-algorithm} shows that the solution space
in \eqref{eq:canonical-system} has binary dimension three and that every
nonzero solution is invertible. By \eqref{eq:cross-covariance},
$Q_{H_{p_0}}(x)=D Q_{H_s}(Tx)$.
Using \eqref{eq:cross-change} gives \eqref{eq:canonical-equivalence} in
the direct case. In the inverse case,
\eqref{eq:inverse-map-equivalence} gives
$Q_{H_s}(Tx)=Q_{H_p}(H_p^{-1}Tx)$, which gives exactly the second
line of \eqref{eq:canonical-witness}. This also proves invertibility.

Computing one $f_p$ requires $m$ constant-size matrix products and $m$
entrywise Frobenius applications. Each application uses at most $m$
squarings per entry, giving $O(m^2)$ field operations. For $7\mid m$,
the extensions used have binary degree at most $3m$.
Lemma~\ref{lem:deterministic-norm} uses $O(\log m)$ root extractions,
and therefore $O(m^3\log m)$ bit operations with schoolbook arithmetic;
the norm exponentiations and optional cubic splitting also have polynomial
cost. The cyclic conversion uses a constant number of Frobenius applications
and three-by-three matrix operations. Constructing the $9m$ binary
columns of \eqref{eq:canonical-system} costs $O(m^2)$ field operations,
and elimination costs $O(m^3)$ bit operations with $O(m^2)$ bits of
storage. The remaining inversions cost $O(m)$ field multiplications by
exponentiation. Since a field multiplication costs $O(m^2)$ bit operations,
the existing label and witness arithmetic costs at most $O(m^4)$ bit
operations. These bounds imply the displayed conservative bound.
The field polynomial $h$ is supplied; its construction is not part of
the complexity claim. No primitive element or discrete logarithm is needed.
\end{proof}

The representatives depend deterministically on $(m,h,k_0,g)$; they are
not defined as lexicographic minima of the parameter triples.
The algorithm applies to the parameter
representation \eqref{eq:family}; its correctness concerns ordinary CCZ
equivalence, by Theorem~\ref{thm:main}.

\section{Conclusion}
For every $m>1$, we classify all admissible coefficients and Frobenius
exponents of the G\"olo\u{g}lu--K\"olsch construction. Each exponent
represents two $K$-linear classes of $(2^{3m}-1)/7$ triples; pairing
inverse exponents and reciprocal cubics gives $\varphi(m)$ global CCZ
classes. We give explicit equivalence maps to Li--Kaleyski representatives
when $7\nmid m$, and construct admissible replacements and equivalence
maps when $7\mid m$. A deterministic norm-preimage construction makes
canonical representative selection and witness recovery polynomial-time
for all degrees in a supplied field encoding.
\bibliographystyle{splncs04}
\bibliography{references}

@article{CarletCharpinZinoviev1998,
  author  = {Claude Carlet and Pascale Charpin and Victor Zinoviev},
  title   = {Codes, Bent Functions and Permutations Suitable for {DES}-Like Cryptosystems},
  journal = {Designs, Codes and Cryptography},
  volume  = {15},
  number  = {2},
  pages   = {125--156},
  year    = {1998},
  doi     = {10.1023/A:1008344232130}
}

@article{CanteautCouvreurPerrin2022,
  author  = {Anne Canteaut and Alain Couvreur and L{\'e}o Perrin},
  title   = {Recovering or Testing Extended-Affine Equivalence},
  journal = {IEEE Transactions on Information Theory},
  volume  = {68},
  number  = {9},
  pages   = {6187--6206},
  year    = {2022},
  doi     = {10.1109/TIT.2022.3166692}
}

@misc{GologluKolschTriprojective2026,
  author        = {Faruk G{\"o}lo\u{g}lu and Lukas K{\"o}lsch},
  title         = {Triprojective Almost Perfect Nonlinear Permutations and Functions},
  year          = {2026},
  eprint        = {2605.17545},
  archiveprefix = {arXiv},
  primaryclass  = {math.CO},
  note          = {Version 1},
  url           = {https://arxiv.org/abs/2605.17545v1}
}

@article{Yoshiara2012,
  author  = {Satoshi Yoshiara},
  title   = {Equivalences of Quadratic {APN} Functions},
  journal = {Journal of Algebraic Combinatorics},
  volume  = {35},
  number  = {3},
  pages   = {461--475},
  year    = {2012},
  doi     = {10.1007/s10801-011-0309-1}
}

@article{Kaleyski2022,
  author  = {Nikolay Kaleyski},
  title   = {Deciding {EA}-Equivalence via Invariants},
  journal = {Cryptography and Communications},
  volume  = {14},
  number  = {2},
  pages   = {271--290},
  year    = {2022},
  doi     = {10.1007/s12095-021-00513-y}
}

@article{Dempwolff2010,
  author  = {Ulrich Dempwolff},
  title   = {On Irreducible Semilinear Transformations},
  journal = {Forum Mathematicum},
  volume  = {22},
  number  = {6},
  pages   = {1193--1206},
  year    = {2010},
  doi     = {10.1515/forum.2010.062}
}

@article{GologluKolschBiprojective2025,
  author  = {Faruk G{\"o}lo\u{g}lu and Lukas K{\"o}lsch},
  title   = {Equivalences of Biprojective Almost Perfect Nonlinear Functions},
  journal = {Combinatorial Theory},
  volume  = {5},
  number  = {3},
  year    = {2025},
  note    = {Paper No. 7, 29 pp.},
  doi     = {10.5070/C65365555},
  url     = {https://escholarship.org/uc/item/8h67z6mh}
}

@article{ShiPengKanGao2025,
  author  = {Chenmiao Shi and Jie Peng and Haibin Kan and Jinjie Gao},
  title   = {On {CCZ}-Equivalence of Two New {APN} Functions in Trivariate Form},
  journal = {Designs, Codes and Cryptography},
  volume  = {93},
  number  = {10},
  pages   = {4595--4625},
  year    = {2025},
  doi     = {10.1007/s10623-025-01699-5}
}

@article{LiKaleyski2024,
  author  = {Kangquan Li and Nikolay Kaleyski},
  title   = {Two New Infinite Families of {APN} Functions in Trivariate Form},
  journal = {IEEE Transactions on Information Theory},
  volume  = {70},
  number  = {2},
  pages   = {1436--1452},
  year    = {2024},
  doi     = {10.1109/TIT.2023.3313148},
  url     = {https://eprint.iacr.org/2022/1522}
}

@misc{BartoliStanica2026,
  author        = {Daniele Bartoli and Pantelimon St{\u a}nic{\u a}},
  title         = {Infinite Families of {APN} Permutations in Constrained Trivariate Classes over $\mathbb{F}_{2^m}$},
  year          = {2026},
  eprint        = {2603.15146},
  archiveprefix = {arXiv},
  primaryclass  = {math.NT},
  note          = {Version 1},
  url           = {https://arxiv.org/abs/2603.15146v1}
}

@article{GologluClassification2023,
  author  = {Faruk G{\"o}lo\u{g}lu},
  title   = {Classification of $(q,q)$-Biprojective {APN} Functions},
  journal = {IEEE Transactions on Information Theory},
  volume  = {69},
  number  = {3},
  pages   = {1988--1999},
  year    = {2023},
  doi     = {10.1109/TIT.2022.3220724},
  url     = {https://arxiv.org/abs/2206.00958}
}

@article{KolschExtension2024,
  author  = {Lukas K{\"o}lsch},
  title   = {On a Recent Extension of a Family of Biprojective {APN} Functions},
  journal = {Finite Fields and Their Applications},
  volume  = {99},
  pages   = {102494},
  year    = {2024},
  doi     = {10.1016/j.ffa.2024.102494},
  url     = {https://arxiv.org/abs/2402.11329}
}

@article{KaspersZhouLower2022,
  author  = {Christian Kaspers and Yue Zhou},
  title   = {A Lower Bound on the Number of Inequivalent {APN} Functions},
  journal = {Journal of Combinatorial Theory, Series A},
  volume  = {186},
  pages   = {105554},
  year    = {2022},
  doi     = {10.1016/j.jcta.2021.105554},
  url     = {https://arxiv.org/abs/2002.00673}
}

@article{KaspersZhouExponential2021,
  author  = {Christian Kaspers and Yue Zhou},
  title   = {The Number of Almost Perfect Nonlinear Functions Grows Exponentially},
  journal = {Journal of Cryptology},
  volume  = {34},
  number  = {1},
  note    = {Article 4},
  year    = {2021},
  doi     = {10.1007/s00145-020-09373-w},
  url     = {https://arxiv.org/abs/2004.11896}
}

@inproceedings{Shoup1991,
  author = {Victor Shoup},
  title = {A Fast Deterministic Algorithm for Factoring Polynomials over Finite Fields of Small Characteristic},
  booktitle = {Proceedings of the 1991 International Symposium on Symbolic and Algebraic Computation},
  year = {1991},
  pages = {14--21},
  url = {https://www.shoup.net/papers/quadfactor.pdf}
}
\end{document}